\documentclass{article}
\usepackage[utf8]{inputenc}
\usepackage{amsmath, amssymb, amsthm, amscd}
\usepackage{graphicx}
\usepackage{hyperref}
\usepackage{geometry}

\usepackage[all]{xy}

\newtheorem{theorem}{Theorem}[section]

\newtheorem{proposition}[theorem]{Proposition}

\newtheorem{definition}[theorem]{Definition}
\newtheorem{remark}[theorem]{Remark}

\newtheorem{question}[theorem]{Question}

\newcommand{\bin}{\{0,1\}^*} 
\newcommand{\T}{\mathtt{T}}
\newcommand{\TB}{\mathtt{TB}}
\newcommand{\st}{\mathtt{stab}}
\newcommand{\I}{\mathtt{I}}
\newcommand{\m}{m^U}
\newcommand{\ms}{m^U_G}
\newcommand{\fix}{\mathtt{fix}}
\newcommand{\s}{\mathtt{sec}}

\title{A Geometric Coding Theorem}
\author{Romie Banerjee}
\date{}
\begin{document}
\maketitle

\begin{abstract}
This paper establishes a direct analogue of the classical Coding Theorem in the setting of symmetry groups. We consider computable bijections on the set of binary strings and define the symmetry prior of a string $x$ as the probability that a randomly chosen symmetry from a given group $G$ has $x$ as its unique fixed point. We show that for any fix-retractable symmetry group $G$, a group admitting a computable section that selects an isolating symmetry for every string, the symmetry prior is a universal lower semi-computable semi-measure. In this case, the Geometric Coding Theorem holds:
$$
-\log_2 m_G(x) = K(x) + O(1)
$$

where $K(x)$ is the prefix Kolmogorov complexity of $x$. This result is a coding-theoretic restatement of a theorem of Trejo,
Kreinovich and Longpr\'e, who showed that the complexity of describing a
string by a symmetry with that string as its unique fixed point equals
its Kolmogorov complexity. Our contribution is to recast it in terms of
algorithmic probability and to treat the symmetry group as a parameter,
identifying fix-retractability as exactly the condition under which
symmetry complexity collapses onto Kolmogorov complexity.


\end{abstract}

\section{Introduction} 
\subsection{Program-Size Complexity}
The concept of algorithmic complexity defined by Solomonoff, Kolmogorov, Chaitin and Levin is a measure that quantifies the algorithmic randomness of a binary string. Formally, the algorithmic complexity, denoted by $K$ of a binary string $x$ is the length of the shortest computer program $p$ running on a universal Turing machine $U$  that generates the string $x$ as an output and halts [\cite{LV} Ch. 2].
\begin{equation}
    K_U(x) = \min\{\ell(p) \mid U(p) = x\}
\end{equation}

The \emph{invariance theorem} says that the definition of $K$ above is independent of the choice of the universal Turing machine $U$. For a a different choice $V$, the $K$ values differ by a constant term independent of $x$. The constant term comes from the cost of simulating $U$ by $V$ and vice-versa. 

The function $K$ that takes in $x$ and gives the length of the shortest program producing $x$ is uncomputable. This is due to the halting problem, given that one cannot always find  the shortest problem in finite time without having to run all programs. This means having to wait forever in case they never halt. The function $K$ is however \emph{upper semi-computable}, i.e. it can approximated from above by computable functions. This is because the set of computer programs that output $x$ and halt is recursively enumerable.

\subsection{Algorithmic Probability} The ordinary probability of production of a binary string $x$  among all possible $2^n$ strings of length $n$ is given by $P(x) = 1/2^n$. The concept of algorithmic probability replaces the random production of outputs with the random production of programs that produce the output. 

The algorithmic probability (or Solomonoff prior or Levin's semi-measure) $m(x)$ of a binary string $x$ is the sum over all prefix-free programs $p$  for which a universal Turing prefix machine $U$ running $p$ outputs $x$ and halts. The replaces the length of $x$ with the length of the program $p$ that produces $x$. It is thus a measure that estimates the probability of a random program $p$  producing $x$ when run on $U$. [\cite{LV} Ch. 3]
\begin{equation}
    \m(x) = \sum_{p \mid U(p) = x} \frac{1}{2^{\ell(p)}}
\end{equation}
The algorithmic probability is independent of the choice of universal prefix machine (up to a constant term) and as a function from binary strings to $\mathbb{R}$ is \emph{lower semi-computable}. 

\subsection{Coding Theorem}
There is a direct and miraculous connection between the algorithmic complexity and algorithmic probability of binary strings. 

The negative logarithm of the algorithmic probability $-\log_2 \m(x)$ gives a complexity measure. By this measure the complexity of a string is directly related to the frequency of production of the string in a universal prefix machine. More programs imply lower complexity. This is related to the algorithmic complexity as the shortest program producing $x$ on $U$ will contribute to the largest term contributing to $\m(x)$. The inequality $-\log_2 \m(x) \leq K_U(x)$ is obvious. 

The algorithmic coding theorem says that something much stronger is true. 
\begin{theorem} [\cite{LV} Thm.4.3.3 ]
    \label{thm:coding} 
    The algorithmic coding theorem,
    \begin{equation}
    K_U(x) = -\log_2 \m(x) + O(1)
    \end{equation}
\end{theorem}
The theorem establishes that, if there are many programs producing a string then there is also a short program producing the string ; and strings of lower frequency have higher Kolmogorov complexity. The key idea behind the proof comes down to the observation that if there are $k$ prefix programs of length $n$ producing $x$ on $U$, then there is a prefix program of length $n - \log_2 k + O(1)$ that produces $x$ on $U$.

There is another way of state the coding theorem. The semi-measure on bit strings defined by $P(x) = 1/2^{K(x)}$ is a \emph{universal} lower semi-computable semi-measure.

\subsection{Geometric Coding Theorem}
A symmetry is a computable bijection of the set of binary strings. If a symmetry has a string $x$ as its only fixed point, it can be framed as program that produces the string $x$ and halts. 

\begin{enumerate}
    \item[]
    \emph{Q: Do the frequency of symmetry programs producing $x$ have a direct relationship with the Kolmogorov complexity of $x$ in the same way the frequency of ordinary programs does (via the classical coding theorem)? }
\end{enumerate}

Let the \emph{symmetry prior} be the probailiy of producing a bit string by randomly selecting a symmetry program. The symmetry prior $\ms(x)$ of a bit string $x$ is the sum over all prefix-free programs symmetry programs $p$ for which a universal Turing prefix machine $U^G$ running $p$ outputs $x$ and halts.
\begin{equation}
    \ms(x) = \sum_{p \mid U^G(p) = x} \frac{1}{2^{\ell(p)}}
\end{equation}

In this paper we lay down a condition on the group of symmetries so that the coding theorem with the symmetry prior is true. 

\begin{theorem}
    \label{thm:gct}
    Let $G$ be a group of computable bijections of $\bin$ with following property: there exists a subgroup $H \leq G$ which can isolate all bit strings and the symmetries $H$ and isolating symmetries $\I_H$ are recursively enumerable. Then the following (equivalent) statements are true:
    \begin{enumerate}
        \item $K_U(x) = -\log_2 \ms(x) + O(1)$.
        \item The symmetry prior dominates the Solomonoff prior: $\exists C$ s.t $C.\ms(x) \geq \m(x), \forall x \in \bin$. 
        \item The symmetry prior is a universal lower semi-computable semi-measure on finite bit strings
    \end{enumerate}
\end{theorem}

\subsection{Results Dependencies}


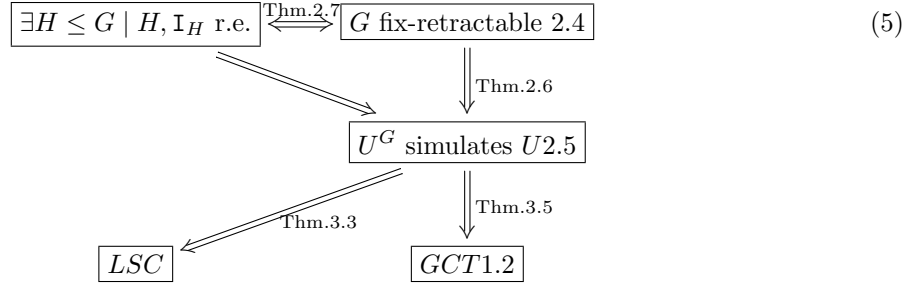
\begin{equation}
\xymatrix{
\boxed{\exists H \leq G \mid  H, \I_H \text{ r.e.}} \ar@{=>}[dr] \ar@{<=>}[r]^{\text{Thm.}\ref{thm:alg}} &\boxed{G \text{ fix-retractable } \ref{def:fix-ret}}\ar@{=>}[d]^{\text{Thm.}\ref{thm:simul}}\\
&\boxed{U^G \text{ simulates } U \ref{def:simul}} \ar@{=>}[d]^{\text{Thm.}\ref{thm:universal}} \ar@{=>}[dl]^{\text{Thm.}\ref{thm:lsc}}\\
\boxed{LSC}  &\boxed{GCT \ref{thm:gct}} \\
}
\end{equation}

\subsection{Relation to prior work.}
The central fact behind this paper is due to Trejo, Kreinovich and
Longpr\'e~\cite{TKL}. They defined the complexity of a string as the
length of the shortest description of a symmetry that has the string as
its unique fixed point, and showed that this \emph{group-symmetric
complexity} coincides with Kolmogorov complexity up to an additive
constant. Both directions of their argument appear here. The isolator
$s_x$ of Theorem~3.6 is their construction, and the recovery of $x$ by
searching for the fixed point is their converse. Theorem~3.5 should
therefore be read as a coding-theoretic restatement of their result: we
transfer it from plain complexity to prefix complexity and algorithmic
probability, where it takes the form of a coding theorem, $-\log
m_G(x)=K(x)+O(1)$.

What the reformulation adds is a change of viewpoint rather than a new
equivalence. First, replacing shortest descriptions by a prior,
$m_G(x)$, the probability that a random program describes a symmetry
isolating $x$, places symmetry inside Solomonoff's theory of inductive
inference, where it can be compared with other priors and used for
prediction. Second, making the group $G$ a parameter turns a single
equivalence into a family of complexity measures, one for each symmetry
group. Fix-retractability is identified as exactly the condition under
which this family collapses onto $K$, and both directions of the
argument are seen to be needed only in one direction each: recovering
$x$ from an isolator works for every $G$, while producing an isolator
from $x$ is what a section supplies. This isolates the precise point
where group structure could make symmetry complexity differ from
Kolmogorov complexity. That point is the subject of the open problems in
Section~\ref{sec:open}.

\section{Coding with symmetries}

\subsection{Symmetries}

The set of total recursive functions $\bin \to \bin$ is denoted by $\T$. The \emph{group} of total recursive \emph{bijective} functions is denoted $\TB$. Elements of $\TB$ are called \emph{symmetries}.
\begin{equation}
    \T = \{ f : \bin \to \bin \mid f(x) \text{ halts } \forall x \}
\end{equation}
\begin{equation}
    \TB = \{ f \in \T \mid f \text{ is a bijection} \}
\end{equation}

Given $x \in \bin$, the set of symmetries that fix $x$ is a subgroup of $\TB$ denoted by $\st(x)$.
\begin{equation}
    \st(x) = \{f \in \TB \mid f(x) = x\} \leq \TB
\end{equation}

A symmetry $f \in \st(x)$ \emph{isolates} $x$ when $x$ is the only fixed point of $x$. The set of isolating symmetries for $x$ is denoted by $I_x$.
\begin{equation}
    \I(x) = \{ f \in \TB \mid f(x) = x \text{ and } f(y) \neq y,  \forall y \in \bin \} \subseteq \st(x)
\end{equation}

Let $\I$ denote the set of all isolating symmetries in $\TB$. 
\begin{equation}
    \I = \{ f \in \TB \mid |\fix(f)| = 1 \} = \coprod_{x \in \bin} \I(x)
\end{equation}

Let $G \leq \TB$ be a subgroup of total bijections of $\bin$. Subgroups of $\TB$ will be called \emph{symmetry groups}. Given $x \in \bin$, the elements of $G$ that fix $x$,
\begin{equation}
    \st_G(x) = \{g \in G \mid g(x) = x \} \leq G 
\end{equation}
The set of isolating symmetries in $G$ for $x$ is denoted by $\I_G(x)$.
\begin{equation}
        \I(x) = \{ g \in G \mid g(x) = x \text{ and } g(y) \neq y,  \forall y \in \bin \} \subseteq \st_G(x)
\end{equation}
The set of all isolating symmetries in $G$ is denoted $\I_G$.
\begin{equation}
    \I_G = \{ g \in G \mid |\fix(g)| = 1 \} = \coprod_{x \in \bin} \I_G(x)
\end{equation}

\subsection{Algorithmic properties of symmetry groups}

The group of total recursive bijections \emph{$\TB$ is not r.e.}. To enumerate all computable bijections, one needs to enumerate all programs that compute total bijective functions. This requires checking totality and bijectivity, which are $\Pi_2$-complete.

Since $\TB$ is not r.e., there is no effective enumeration of $\st(x)$ either. If we make the assumtion however that $G$ is a r.e. subgroup of $\TB$, the stabiliser group $\st_G(x)$ is r.e. One can enumerate through $G$ and check if $g \in G$ fixes $x$, provided that the action of the group on bit strings is computable . 

Under the assumption that $G$ is r.e., the subset $\I_G(x)$ is co-r.e. To enumerate $\I_x$, one would enumerate elements of $\st_G(x)$ to check that they have no fixed points other than $x$. This is a co-r.e. condition (finding the second fixed point is r.e., proving none exists is co-r.e.). 

\begin{proposition}
    \label{prop:symmalg}
    Let $G$ be a r.e. symmetry group. Then for all $x \in \bin$, the stabilizer subgroups $\st_G(x)$ are r.e.
    The point-wise isolators and total isolator $\I_G(x)$ and $\I_G$ are co-r.e.    
\end{proposition}

A prefix universal computer is a partial recursive function $U: 
\bin \times \bin \to \bin$ so that $U(p,x) = y$ where $p$ is prefix-free code for the function $U(p,\cdot) : \bin \to \bin$. 


A program encoding a symmetry is a prefix code $p$ so that $U(p, \cdot)$ is a symmetry. We call this a \emph{symmetry program}. The set of all symmetry programs is denoted by $U^{\TB}$.
\begin{equation}
    U^{\TB} = \{ p \mid U(p, \cdot) \in \TB \}
\end{equation}

By the discussion there $U^{\TB}$ is not a recursively enumerable language. Let $G \subset \TB$ be a r.e. subset. The set of symmetry programs for symmetries in $G$
\begin{equation}
    U^G = \{ p \mid U(p, \cdot) \in G \}
\end{equation}
is r.e. Let us denote the universal computer for $G$-symmetries by the same name.

\begin{proposition}
    \label{prop:symmcomp}
    Let $G$ be a r.e. symmetry group. There is an effective enumeration of $G$-symmetries. The universal computer for this is a p.r. function 
    \begin{equation}
    U^G: \bin \times \bin \to \bin 
    \end{equation}
    such that $U^G(p, \cdot) \in G$ when $p$ is halting. 
\end{proposition}

A symmetry program \emph{computes} $x$ when it is a program for an isolating symmetry for $x$. The set of all symmetry programs computing $x$ is denoted by $U^{\TB}_x$.
\begin{equation}
    U^{\TB}_x = \{ p \mid U(p, \cdot) \in \I_x\}
\end{equation}

The set of all symmetry programs in $U^G$ computing $x$ is denoted by $U^G_x$.
\begin{equation}
    U^{G}_x = \{ p \mid U^G(p, \cdot) \in \I_x\}
\end{equation}


\begin{remark}
Note that $U^{\TB}$ is not a universal computer for total recursive bijections. The set $\TB$ is not recursively enumerable. There is no Turing machine that accepts only the symmetry prefix programs in $U$.  
\end{remark}

\subsection{Simulations}
The fixed-point set of computable symmetries offers a mapping from $G \leq \TB$ to subsets of bit strings $G \to \mathcal{P}(\bin)$. This is non-computable in general. Restriction to isolating symmetries gives a computable map,
\begin{equation}
    \fix: \I_G \to \bin.
\end{equation} 
\begin{definition}
    \label{def:fix-ret}
    An $I$-section of $G \leq \TB$ is a computable right-inverse of $\fix$, i.e. a computable map 
    \begin{equation}
        \s: \bin \to \I_G
    \end{equation}
    such that $\fix\circ\s = \text{id}_{\bin}$. The existence of the section means $\I_G(x)$ is non-empty $\forall x \in \bin$.

    A symmetry group $G$ is \emph{fix-retractable} if it admits an $I$-section. 
\end{definition}



The existence of a $I$-section implies an important property of the program spaces. Any ordinary program for outputting a bit string can be simulated by a symmetry program of the same length (up to additive constant term). The $\fix$ map implies a similar property in the other direction, any symmetry program that isolates a bit string can be simulated by an ordinary program of the same length (up to a additive constant). The idea of codes being simulatable by codes of another type of the same length is formalized in the definition below. The subsequent proposition proves the connection claimed. 

\begin{definition}
    \label{def:simul}
    $U^{G}$ simulates $U$ if there is a injective computable function $\Phi:\bin \to \bin$ that maps programs computing $x$ to symmetry programs computing $x$.
    \begin{equation}
        \Phi(U_x) \subset U^G_x, \forall x \in \bin
    \end{equation}
    
    and $\ell(\Phi(p)) = \ell(p) + O(1)$ for all $p \in U_x$. 
    
    $U$ simulates $U^G$ if there is a injective computable function $\Psi:\bin \to \bin$ that maps symmetry programs to programs computing $x$ 
    \begin{equation}
        \Psi(U^G_x) \subset U_x, \forall x \in \bin
    \end{equation}
    and $\ell(\Psi(q)) = \ell(q) + O(1)$ for all $q \in U^G_x$.
\end{definition}

\begin{theorem}
    \label{thm:simul}
    If $G$ is a fix-retractable symmetry group, then $U$ and $U^G$ can simulate each other.
\end{theorem}

\begin{proof}
    \begin{enumerate}
    \item ($U$ simulates $U^G$):
    The map $\fix: \I_G \to \bin$ lifts to a computable map between program spaces:
    \begin{equation}
        U_{\fix}: \{q \mid U^G(q) \text{ halts and} \in \I_G\} \to \{ p \mid U(p) \text{ halts }\}
    \end{equation}
    such that $\forall x \in \bin$, $U_{\fix}(U^G_x) \subset U_x$. The mapping $U_{\fix}$ is defined as follows: given a program $q$ for an $x$-isolator, run $q$, get the isolator $g \in \I_G(x)$, compute the fixed point by computing $g(x)$ through enumerating $x \in \bin$. This will stop eventually because $g$ is an isolator for $x$. The program $U_{\fix}(q)$ is therefore a program that produces $x$. The prefix code for $U_{\fix}(q)$ is obtained by concatenating the prefix code with $q$ with the code for $U_{\fix}$ plus a constant overhead code. This is injective and the length,
    \begin{equation}
        \ell(U_{\fix}(q)) \leq \ell(q) + O(1).
    \end{equation}

    \item ($U^G$ simulates $U$):
    The section $\s: \bin \to \I_G$ lifts to a computable map between program spaces:
    \begin{equation}
        U_{\s}: \{ p \mid U(p) \text{ halts }\} \to \{q \mid U^G(q) \text{ halts and} \in \I_G\}
    \end{equation}
    such that $\forall x \in \bin$, $U_{\s}(U_x) \subset U^G_x$. The mapping $U_{\s}$ is defined as follows: given a program $p$ that produces $x$, run $p$, get $x$, apply $\s$ to get $\s(x) \in \I_G(x)$. The program $U_{\s}(p)$ is program for an isolator for $x$ by concatenating the prefix code for $p$ with the constant length prefix code for $U_{\s}$. This makes $U_{\s}$ injective and the length,
    \begin{equation}
        \ell(U_{\s}(p)) \leq \ell(p) + O(1).
    \end{equation}

    \end{enumerate}
\end{proof}

We give an alternate characterization of fix-retractable symmetry groups. The presence of an $I$-section implies and is implied by the existence of a subgroup of $G$ which is r.e. and for which the set of isolators is also r.e.

\begin{theorem} 
    \label{thm:alg}
    The following properties of a symmetry group $G \leq \TB$ are equivalent.
    \begin{enumerate}
        \item $G$ is fix-retractable 
        \item $\exists H \leq G$ such that
        \begin{enumerate}
            \item $H$ can isolate all $x \in \bin$
            \item $H$ and $\I_H$ are r.e.
        \end{enumerate}
    \end{enumerate}
    Therefore for a symmetry group $G$ satisfying 2, $U^G$ can simulate $U$.
\end{theorem}

\begin{proof}
\begin{enumerate}
    \item[]
    \item ($\Leftarrow$)
    Let $U^H$ and $U^{\I_H}$ be effective enumerations of $H$ and $\I_H$. Define the I-section of $G$, $\s: \bin \to \I_G$ as follows: 

    Enumerate isolators in $H$ by running $U^{\I_H}$. For every $h \in \I_H$ check if it belongs to $\I_H(x)$ by computing $h(x)$. This will eventually halt because of the first assumption about $H$. 

    \item ($\Rightarrow$)
    Let $\s: \bin \to \I_G$ be an I-section for $G$. Consider the induced map on programs (see proof of \ref{thm:simul})
    \begin{equation}
        U_{\s}: \{ p \mid U(p) \text{ halts }\} \to \{q \mid U^G(q) \text{ halts and} \in \I_G\}
    \end{equation}
    which satisfies $\s(U_x) \subseteq U^G_x$ for all $x \in \bin$. 

    Consider the set of isolators in $G$ defined by the image of $\s$:
    \begin{equation}
        S = \{ U(q, \cdot) \mid q \in \s(U_x) \} \subseteq \I_G(x)
    \end{equation}
    Let $H := \langle h : h \in S \rangle$ be the subgroup of $G$ generated by $S$. 

\end{enumerate}
\end{proof}

\section{Algorithmic Probability from Symmetries}

\subsection{Symmetry Priors}




\begin{definition}
The $G$-\emph{symmetry prior} of a string $x$ with respect to a computer $U$ is the probability that a random symmetry program computes $x$. 
\begin{equation}
    \ms(x) = \sum_{p \in U_x^G} \frac{1}{2^{\ell(p)}}
\end{equation}
\end{definition}

The Solomonoff Prior $\m(x)$ is \emph{lower semi-computable}. The proof relies on the fact that the set of programs $U_x$ is recursively enumerable. The direct analogy fails to work with the symmetry prior since the set of symmetry programs computing $x$ is not r.e.

In the sequel we show that the symmetry prior is equal to the Solomonoff prior up to a constant term when every program can be \emph{simulated} by a symmetry program. But first, we have to make the idea of simulation precise. 

\subsection{Universal Semi-measures}

\begin{definition}
A function $\mu: \bin \to \mathbb{R}$ is a probability semi-measure if 
\begin{enumerate}
    \item $\mu(\epsilon) \leq 1$
    \item $\mu(x) \geq \mu(x0) + \mu(x1)$
\end{enumerate}
A semi-measure is lower semi-computable if the function $\mu$ is lower semi-computable. 

Let $\mathcal{M}$ be a class of discrete semi-measures. A semi-measure $P_0$ is \emph{universal} (or maximal) for $\mathcal{M}$ if $P_0 \in \mathcal{M}$, and for all $P \in \mathcal{M}$ there exists a constant $c_P$ such that for all $x \in \bin$ we have $c_PP_0(x) \geq P(x)$, where $c_P$ possibly depends on $P$ but not on $x$. 
\end{definition}

The Solomonoff prior $\m(x)$ and $2^{-K(x)}$ are lower semi-computable semi-measures. The condition of being semi-measures is a consequence of teh Kraft inequality of prefix codes. Lower semi-computability of both semi-measures come from the fact that $U_x$ is recursively enumerable. 

The prior $\m(x)$ is universal for the following reason. For any lower semi-computable semi-measure $\mu$, there is a computable enumeration of programs that realize $\mu$. The universal machine $U$ can simulate all such programs. The sum $\m(x)$ includes all programs, so it dominates each $\mu$ up to a constant (the weight assigned to the program that enumerates $\mu$).

The classical coding theorem implies that $2^{-K(x)}$ is universal: since $\m(x) \asymp 2^{-K_U(x)}$, and $\m(x)$ is universal, $2^{-K_U(x)}$ inherits universality.

\begin{proposition}
    \label{thm:lsc}
    If $G$ is fix-retractable, the symmetry prior $\ms(x)$ is a lower semi-computable semi-measure.
\end{proposition}
\begin{proof}
Compose the two simulations
\begin{equation}
    \Theta := \s \circ \fix : U^G_x \to U^G_x
\end{equation}
This map takes any symmetry program $p \in U^G_x$, converts it into a standard program $\fix(p)$ and converts it back into a symmetry program $\s(\fix(x))$. 
The length of the program $\Theta(p)$ is no longer than $p$ by a constant term
\begin{equation}
    \ell(\Theta(p)) = \ell(\s(\fix(p)) \leq \ell(p) + c_{\Psi} + c_{\Phi}
\end{equation}
Let $S_x = \Theta(U^G_x) \subseteq U^G_x$ be the set of symmetry programs that are in the image of $\Theta$. Two important properties of $S_x$:

\begin{enumerate}
    \item $S_x$ is recursively enumerable. This is so because $U_x$ is recursively enumerable and $\Phi$ is a computable map. Therefore $S_x$ is the image of a r.e set under a computable map. 
    \item $S_x$ is probability dense in $U^G_x$. 
        \begin{equation}
            \sum_{q \in S_x} \frac{1}{2^{\ell(q)}} \geq \sum_{p \in U_x^G} \frac{1}{2^{\ell{\Theta(p)}}} = 2^{-c_{\theta}}\sum_{p \in U_x^G} \frac{1}{2^{\ell(p)}} = 2^{-\Theta} \ms(x)
        \end{equation}
    So $S_x$ carries a constant fraction of the total probability mass of $U^G_x$.
\end{enumerate}
To approximate the symmetry prior from below we only need to enumerate a probability-dense r.e. subset $S_x \subseteq U^G_x$. Since $S_x$ is r.e., we can enumerate its elements and sum their probabilities, giving a non-decreasing sequence that converges to the value at least $2^{-c_{\Theta}} \ms(x)$.
\end{proof}

The proof of lower semi-computability of the symmetry prior \ref{thm:lsc} inspires the following definition. 

\begin{definition}
    For every $x \in \bin$, let $\mathcal{C}_x$ be a set of codes with prefix-free encoding and a probability semi-measure $\mu$. A subset $
    \mathcal{D}_x \subseteq \mathcal{C}_x$ is called \emph{uniformly probability dense} if:
    \begin{enumerate}
        \item $\mathcal{D}_x$ is recursive enumerable
        \item There exists a constant $c > 0$ such that $\mu(\mathcal{D}_x) \geq c.\mu(\mathcal{C}_x)$ for all $x$. 
    \end{enumerate}
\end{definition}

In the context of symmetry programs, $\mathcal{C}_x = U^{\TB}_x$, and a uniformly probability-dense subset $S_x \subset U_x^{\TB}$ provides a lower semi-computable approximation to the symmetry prior $\ms(x)$.

\begin{theorem}
\label{thm:universal}
If $G$ is a fix-retractable symmetry group then the symmetry prior is equal to the Solomonoff prior up to a multiplicative constant. 

\begin{equation}
    \ms(x) \asymp \m(x).
\end{equation}

Therefore, the symmetry prior is a \emph{universal} lower semi-computable semi-measure on binary strings. Also as a consequence of the classical coding theorem
 \begin{equation}
    -\log_2 \ms(x)= K(x) + O(1)
    \end{equation}

\end{theorem}

\begin{proof}
\begin{enumerate} 
    \item[]  
    \item ($\ms(x) \geq C. \m(x)$) 
    \item[] From the simulation $\s: U \to U^G$, every program $p$ for $x$ gives a symmetry program $\s(p)$ with $\ell(\s(p)) = \ell(p) + c_1$, where $c_1$ is the constant overhead of the simulation. Therefore: 
    \begin{equation}
        \ms(x) = \sum_{p \in U_x^G} \frac{1}{2^{\ell(p)}} \geq \sum_{p \in U_x} \frac{1}{2^{\ell({\s(p)})}} = 2^{-c_1}\sum_{p \in U_x} \frac{1}{2^{\ell(p)}} = 2^{-c_1} \m(x)
    \end{equation}
    \item ($\m(x) \geq C. \ms(x)$)
    \item[]
        From the simulation $\fix: U^G_x \to U_x$, every symmetry program $q$ for $x$ gives a program $\fix(q)$ for $x$ with $\ell(\fix(q)) = \ell(q) + c_2$ where $c_2$ is the constant overhead of the simulation $\fix$. Therefore:
        \begin{equation}
        \m(x) = \sum_{p \in U_x} \frac{1}{2^{\ell(p)}} \geq \sum_{q \in U^G_x} \frac{1}{2^{\ell(\fix(q))}} = 2^{-c_2}\sum_{q \in U^G_x} \frac{1}{2^{\ell(q)}} = \ms(x)
        \end{equation}

\end{enumerate}
\end{proof}

\subsection{Examples}
\subsubsection{$G = \TB$}
\begin{theorem}\cite{TKL}
    \label{thm:TB}
    $\TB$ is fix-retractable. Hence the GCT is true for the full symmetry group of all computable symmetries.
\end{theorem}
\begin{proof}An I-section of $\TB$ would be a computable map $\s: \bin \to \I$ such that $\s(x) \in \I_x$. We give an explicit construction. For every $x \in \bin$ let $s_x$ be a $x$ isolator defined as follows:
    \begin{enumerate}
        \item $s_x(x) = x$
        \item $s_x(x^*) = \epsilon$ and $s_x(\epsilon) = x^*$
        \item $\forall y \neq x, s_x(y) = y^*$, where $y^*$ is $y$ with its last bit flipped.
    \end{enumerate}
    The mapping $x \mapsto s_x$ is computable.
\end{proof}

\section{Open problems}
\label{sec:open}

The geometric coding theorem (Thm.\ref{thm:universal}) holds for fix-retractable
groups. There the passage from strings to isolators is supplied by a
computable section, and the passage from isolators back to strings,
by searching for the fixed point, works for every $G$ (Thm.\ref{thm:simul}).
Fix-retractability is therefore exactly the condition under which the
symmetry prior collapses onto the Solomonoff prior. The problems below
concern what happens when this condition fails, and how the theory
interacts with group structure.

Throughout, $G$ is an r.e.\ symmetry group with a fixed effective
enumeration $(g_i)_{i\in\mathbb{N}}$ (Proposition~2.2). Write
$\mathrm{Iso}_G(x)=\{i : g_i\in I_G(x)\}$ and
\[
  K_G(x)=\min\{K(i) : i\in\mathrm{Iso}_G(x)\},
\]
with $K_G(x)=\infty$ if $I_G(x)=\emptyset$. Searching for the fixed point
gives $K(x)\le K_G(x)+O(1)$ for every $G$.

\subsection{The coding theorem without a section}

Without a section, only $-\log m_G(x)\le K_G(x)+O(1)$ is immediate. The
classical proof of the converse uses that the set of programs for $x$
is r.e., whereas $\mathrm{Iso}_G(x)$ is in general only co-r.e.\
(Proposition~2.1). The question becomes whether a set of large
algorithmic probability must contain a simple member. 

\begin{question}
    \label{q:EL}
Let $S_x$ be a finite truncation of $\mathrm{Iso}_G(x)$ carrying at least
half of $m_G(x)$. Is
\[
  K_G(x)\le-\log m_G(x)+O\big(I(S_x:\mathcal H)\big)+O(\log K_G(x))
\]
for every r.e.\ $G$, where $\mathcal H$ is the halting sequence?
Conversely, is there an r.e.\ $G$ with
$\sup_x\big(K_G(x)+\log m_G(x)\big)=\infty$? An affirmative answer to
both would say that a geometric coding theorem survives beyond
fix-retractability exactly to the extent that isolator sets are
stochastic.
\end{question}

\begin{question}
    \label{q:lsc}
Theorem~3.5 identifies $m_G$ with a universal lower semi-computable
semimeasure up to a multiplicative constant. Is $m_G$ \emph{itself}
lower semi-computable, for $G=\mathtt{TB}$ or for any fix-retractable
$G$? Membership of a program in $\mathrm{Iso}_G(x)$ combines a
$\Pi_2$ condition (totality and bijectivity) with a $\Pi_1$ condition
(uniqueness of the fixed point), so we expect a negative answer in
general.
\end{question}




\subsection{From isolators to orbits}

In $\mathtt{TB}$ the isolator sets of any two strings are conjugate:
conjugating by the transposition exchanging $x$ and $y$ maps $I(x)$
bijectively onto $I(y)$. So group-theoretically all strings are equally
isolable. Strings differ only in how \emph{briefly} their isolators can
be described, and that difference is governed by $K(x)$ through the cost
of the conjugating transposition. Isolators are moreover \emph{lossless}
descriptions, so they can recover $K(x)$ but say nothing about the
structure of $x$.

\begin{question}
Develop the \emph{lossy} counterpart, in which $x$ is described by a
group together with the orbit containing it, and the position of $x$ in
the orbit is treated as noise. Here stabiliser sizes, rather than
description lengths of isolators, measure symmetry: a large stabiliser
gives a small orbit and a short two-part code. Which part of the
regularity of $x$ is captured by orbits of a restricted class of
groups? This programme is carried out in the sequel.
\end{question}

\end{document}